%% file: main.tex
\documentclass[10pt,a4paper]{article}
\usepackage[latin1]{inputenc}
\usepackage{amsmath}
\usepackage{amsfonts}
\usepackage{amssymb}
\usepackage{amsthm}
\usepackage[usenames,dvipsnames,svgnames,table]
           {xcolor}

\usepackage{mathrsfs}		

\usepackage{algorithm}
\usepackage{algorithmic}
\usepackage{hyperref}

\usepackage{tikz}

\include{macros}

\providecommand{\ensurecommand}[2]{
\providecommand{#1}{}
\renewcommand{#1}{#2}
}

\providecommand{\integers}{{\mathbb Z}}
\providecommand{\gradient}{\nabla}

\providecommand{\argholder}{{\, \cdot \,}}

\newtheorem{myprob}{Problem}[section]

\include{macros_specific}

\providecommand{\numnodes}{n}
\providecommand{\numedges}{m}
\newcommand{\numpoints}{M}

\providecommand{\coordvar}{y}

\providecommand{\cyclevar}{}
\renewcommand{\cyclevar}{{\mathscr C}}

\providecommand{\polytope}{{\mathbb P}}
\providecommand{\constrmat}{A}
\providecommand{\constrrow}{a}
\providecommand{\constrvec}{b}
\providecommand{\solnvar}{x}
\providecommand{\solnvec}{{\bf\solnvar}}
\providecommand{\costvec}{c}

\providecommand{\pointsone}{S}
\providecommand{\pointstwo}{T}

\title{An $O(\numpoints \log \numpoints)$ Algorithm for Bipartite Matching with Roadmap Distances}
\author{Kyle~Treleaven, Josh~Bialkowski, Emilio~Frazzoli}
\date{}

\begin{document}

\maketitle

\setlength{\marginparwidth}{1.5in}
\newcommand{\ktmargin}[2]{{\color{orange} #1}\marginpar{\small\noindent{\raggedright{\color{orange}[KT]}\color{orange}{#2} \par}}}
\newcommand{\jbmargin}[2]{{\color{Blue} #1}\marginpar{\small\noindent{\raggedright{\color{Turquoise}[JB]}\color{Blue}{#2} \par}}}

\begin{abstract}
An algorithm is presented which produces the minimum cost bipartite matching
between two sets of $\numpoints$ points each,
where the cost of matching two points is proportional to 
the minimum distance by which a particle could reach one point from the other
while constrained to travel on a connected set of curves, or \emph{roads}.
Given any such ``roadmap'',
the algorithm obtains
$O(\numpoints \log \numpoints)$ total runtime in terms of $\numpoints$, which 
is the best possible bound in the sense that any algorithm for minimal matching
has runtime $\Omega(\numpoints \log \numpoints)$.
The algorithm is strongly polynomial and is based on a \emph{capacity-scaling} approach to
the [minimum] convex cost flow problem.
The result generalizes the known $\Theta(\numpoints \log \numpoints)$ complexity
of computing optimal matchings between two sets of points on
(i) a line segment, and (ii) a circle.
\end{abstract}

\section{Introduction}

\providecommand{\points}{P}
\newcommand{\pointvar}{s}
\newcommand{\altpointvar}{t}

\providecommand{\permvar}{}
\renewcommand{\permvar}{\pi}
\providecommand{\Match}{{\mathcal A}}		

\providecommand{\distance}{d}

\newcommand{\length}{L}
\newcommand{\cost}{C}

Given two sets
$\pointsone = (\pointvar_1,\pointvar_2,\ldots,\pointvar_\numpoints)$ and
$\pointstwo = (\altpointvar_1,\altpointvar_2,\ldots,\altpointvar_\numpoints)$,
of $\numpoints$ points each
from a domain $\env$,
a matching is a subset of pairs $\Match \subset \pointsone \times \pointstwo$
such that each point appears in $\Match$ exactly once.
Every matching represents a unique bijective mapping of $\pointsone$ onto $\pointstwo$, and therefore,
is uniquely determined by some permutation $\permvar$ of $(1,2,\ldots,\numpoints)$,
in the sense that $\Match = \{ (\pointvar_k,\altpointvar_{\permvar_k}) \}_{k=1}^\numpoints$.
Given a distance metric $\distance : \env^2 \to \reals_{\geq 0}$,
we define the cost of a match $(\pointvar,\altpointvar) \in \Match$ as
$\distance(\pointvar,\altpointvar)$ and
the cost of the matching $\Match$ as the sum over pairs
$\sum_{(\pointvar,\altpointvar) \in \Match} \distance(\pointvar,\altpointvar)$.
The problem of finding the minimum cost matching, or \emph{minimal} matching, is called
the assignment problem or the bipartite matching problem.

\emph{Literature Review}:
Assignment problems~\cite{burkard2009assignment} have direct applications in domains such as
operations management, computer science, computational biology, and computational music.
For example,
they can be used to represent the problem of scheduling work to assets optimally
(e.g., programs to computer processors, or jobs to machines or workers).
In computational biology, the many-to-one assignment problem on a line is used
to solve the restriction scaffold assignment problem~\cite{colannino2006n},
an important step which is solved millions of times in DNA sequencing.
In computational music, assignment costs are used to measure the similarity of musical rhythms~\cite{diaz2004compas}.
%
Assignment problems also
have numerous uses for solving and approximating other combinatorial optimization problems.
%
For example, bipartite matching dominates the runtime complexity of the so-called LARGEARCS algorithm,
part of the best known constant-factor algorithm~\cite{FHK:S76}
for the \emph{Stacker Crane problem} (SCP).
The SCP is a NP-Hard problem whose objective is
to obtain the smallest tour through many one-way transportation demands.
LARGEARCS has been shown to be asymptotically optimal
for large randomly generated instances (almost surely)~\cite{treleaven2013},
so efficient algorithms for matching are desirable for dealing with large SCP instances.
%

In general, the assignment problems can be is solved in $O(\numpoints^3)$ time
using the classical Hungarian method~\cite{kuhn1955hungarian}.
If the points are in $\reals^2$ and the distance between them is Euclidean,
then there is a class of algorithms~\cite{Agarwal:1995} for the one-to-one assignment, or bipartite matching problem, which
achieve $O(\numpoints^{2+\epsilon})$ time for any $\epsilon>0$.
If the points are on a line, then
there is a trivial $O(\numpoints \log \numpoints)$ algorithm to solve the one-to-one assignment problem,
the optimality of which is proved, e.g., in~\cite{werman1985distance}.
For the case that the points lie on a circle, another $O(\numpoints \log \numpoints)$ algorithm was given
originally by Karp and Li~\cite{karp1975two}, which has been refined somewhat by others, e.g.,~\cite{werman1986bipartite}.
Colannino et. al. provided companion $O(\numpoints \log \numpoints)$ algorithms on the line
for both the many-to-one~\cite{colannino2006n} and
many-to-many~\cite{colannino2007efficient} versions of the assignment problem;
both algorithms are fundamentally based on the original insights developed by Karp about the \emph{circle}, though
they provide non-trivial additional insights themselves.
(The algorithms on the line are actually $O(\numpoints)$ if the points are already sorted; however,
on the circle the runtime remains $O(\numpoints \log \numpoints)$.)
The best known many-to-one and many-to-many algorithms on the circle remain $O(\numpoints^2)$.

\emph{Contribution}:
In this paper, we provide an algorithm which computes the minimal matching between two sets
of $\numpoints$ points each,
on a fixed but arbitrary \emph{roadmap} with $\numedges$ roads and $\numnodes$ vertices.
The runtime of the algorithm is
$O(\numpoints \log \numpoints )
	+ \log \numpoints \times O( \numedges ( \numedges + \numnodes \log \numnodes ) )$, which
for a fixed roadmap is
dominated asymptotically by the term $O(\numpoints \log \numpoints)$.
Such is the best runtime bound achievable in terms of $\numpoints$.
The design of the algorithm is influenced by
the elegant treatment of bipartite matching
on lines and circles in~\cite{werman1986bipartite}, which
are the two roadmap structures one can construct from a single road
($\numedges = 1$).
The crucial component of the algorithm is
the application of a \emph{capacity-scaling} approach for the minimum \emph{convex cost flow problem},
developed, e.g., in~\cite[Sec.~14.5]{ahuja1993network}.

Our algorithm can also provide a speed-up in cases where the roadmap geometry is \emph{implicit}.
Suppose, for example, that
the minimal bipartite matching is desired between $\pointsone$ and $\pointstwo$, where
matches have costs equal to the lengths of shortest paths on
a weighted, undirected graph $(\pointsone \cup \pointstwo \cup V, E)$ with $\numnodes$ vertices and $\numedges$ edges.
A standard approach is to cast such problem as a minimum cost flow problem with unit supplies, which
can be solved in $O(\numedges ( \numedges + \numnodes \log \numnodes ))$ time.
If $\numedges$ is very close to $2\numpoints$, then
such approach represents a factor nearly $\numpoints$ speed-up,
compared, e.g., to the Hungarian method on a metric completion graph.
If additionally most of the $\numedges$ edges have degree $2$, then
using the technique of this paper,
an \emph{additional} factor nearly $\numpoints$ speed up may be possible:
Letting $\numpoints'$ denote (roughly) half the number of degree-$2$ vertices,
one can produce an equivalent matching instance of $2\numpoints'$ points on a roadmap
with $\numnodes'$ vertices and $\numedges'$ edges, where
$\numnodes'$ is the number of \emph{non-} degree-$2$ vertices, and
$\numedges'$ is their total degree.
Then using the algorithm of this paper, the instance can be solved in 
$O(\numpoints' \log \numpoints' )
	+ \log \numpoints' \times O( \numedges' ( \numedges' + \numnodes' \log \numnodes' ) )$
time.

\subsection{Organization}

The rest of the paper is organized as follows:
In Section~\ref{sec:lines and circles} we review the existing results for bipartite matching
on line segments and circles, which
are the two most basic roadmaps.
We state the problem of finding the minimal bipartite matching between points on a general roadmap
rigorously in Section~\ref{sec:problem statement}.
%
%
In Section~\ref{sec:roadmap matching} we generalize the analysis of bipartite matching on line and circles
and introduce an optimal algorithm for bipartite matching on general roadmaps.
We discuss computational complexity in Section~\ref{sec:complexity} and demonstrate that
our algorithm can be computed in $O(\numpoints \log \numpoints)$ time.
We provide some discussion of the results and closing remarks in Section~\ref{sec:discussion}.

\section{Optimal Matching on Lines and Circles} \label{sec:lines and circles}

\providecommand{\nodeset}{V}
\providecommand{\arcset}{A}
\providecommand{\arcvar}{a}

\providecommand{\edgeset}{}
\renewcommand{\edgeset}{E}

\newcommand{\Numarcs}{Z}
\newcommand{\numarcs}{z}
\newcommand{\cumarcs}{F}
\newcommand{\cumval}{f}
\newcommand{\cumvals}{{\bf F}}

\newcommand{\postcumarcs}{H}

\newcommand{\cumpoints}{N}

If the domain $\env$ is a single line segment with
distance defined by $\distance(\pointvar,\altpointvar) = |\altpointvar-\pointvar|$, and
both $\pointsone$ and $\pointstwo$ are \emph{sorted}, then
it is well known (shown, e.g., in~\cite{werman1985distance})
that the minimal matching is determined by the identity permutation $\permvar^1$, i.e.,
$\Match
	= \{ (\pointvar_1,\altpointvar_1),(\pointvar_2,\altpointvar_2),\ldots,(\pointvar_\numpoints,\altpointvar_\numpoints) \}$.
Therefore, if $\pointsone$ and $\pointstwo$ are already sorted, the minimal matching can be transcribed in $\Theta(\numpoints)$ time.
If the points are not sorted, then sorting them first takes $\Theta(\numpoints \log \numpoints)$ time in the worst case.

Such a straightfoward scenario would seem to warrant little additional discussion, but
\cite{werman1986bipartite} presents a characterization of the 
\emph{cost} of the minimal matching, which becomes useful in various extensions:
Let $[0,\length)$ be the interval containing $\pointsone$ and $\pointstwo$, and
let us orient all matches $(\pointvar,\altpointvar)$
from $\pointvar$ to $\altpointvar$.
If $\pointvar < \altpointvar$ we say the match is \emph{forward}, because
the orientation of the match is in the direction of the positive coordinate axis;
if $\pointvar > \altpointvar$, then we say the match is \emph{backward}.
Let $\postcumarcs(y)$ denote
the total number of forward matches crossing a coordinate $y$ \emph{minus}
the total number of backward matches crossing $y$.
%
A matching $\Match$ is called \emph{unidirectional} if at every such coordinate,
all of the matches crossing the point have the same orientation.
\begin{lemma} \label{lemma:minimal is unidirectional}
Every \emph{minimal} matching is unidirectional.
\end{lemma}
\begin{proof}
A simple proof is given in Lemma~3 of~\cite{werman1986bipartite}, but
we provide another proof that generalizes readily to every scenario in this paper.
The proof is by contradiction:
Assume that $\Match^*$ is a minimal match but there is some coordinate $\coordvar$
crossed by $\numpoints^+ > 0$ forward matches and $\numpoints^- > 0$ backward matches.
Consider the case that there are no points at $\coordvar$.
If $\coordvar$ has points, then we simply ignore all the matches that start or end there (they do not \emph{cross}).
With simliar justification, assume that $\coordvar$ is an interior point.
Then there is a neighborhood $(\coordvar-\epsilon,\coordvar+\epsilon)$ for some $\epsilon > 0$ which
contains no points and
is crossed by exactly the matches through $\coordvar$.
Choose one of the forward matches and one of the backward matches and exchange their endpoints in $\pointstwo$.
The resulting matching has cost at least $4\epsilon$ less than the \emph{minimal} matching $\Match^*$.
\end{proof}

Note by Lemma~\ref{lemma:minimal is unidirectional} that
the total number of matches crossing a point $\coordvar$
is $|\postcumarcs(\coordvar)|$
for the minimal matching $\Match^*$;
if $\postcumarcs(\coordvar) > 0$ then they are forward matches;
otherwise, they are backward matches.
Thus, the total length of all matches, i.e.,
the cost of the optimal matching, can be written as
\begin{equation} \label{eq:matching cost on line}
	\int_0^\length | \postcumarcs(y) | \ dy.
\end{equation}

Let $\cumpoints_\pointsone(y)$ denote the number of points $\pointvar \in \pointsone$ such that $\pointvar < y$,
i.e., $| \pointsone \cap [0,\coordvar) |$,
let $\cumpoints_\pointstwo(y)$ denote the number of points $\altpointvar \in \pointstwo$ such that $\altpointvar < y$,
i.e., $| \pointstwo \cap [0,\coordvar) |$, and
let $\cumarcs(y) \doteq \cumpoints_\pointsone(y) - \cumpoints_\pointstwo(y)$.
Note that since $|\pointsone| = |\pointstwo|$ we have $\cumarcs(0) = \cumarcs(\length) = 0$.
As argued in~\cite{werman1986bipartite},
$\cumarcs(y) = \postcumarcs(y)$ on a line segment, because
whenever we cross a point $\pointvar \in \pointsone$ in the positive direction,
either a forward match is beginning, or a backward match is ending;
either way, both $\postcumarcs$ and $\cumarcs$ increase by $1$.
(Whenever we encounter a point $\altpointvar \in \pointstwo$, both $\postcumarcs$ and $\cumarcs$ decrease by $1$.)
Therefore, the cost of the minimal matching can be computed \emph{without} producing one, by
computing $\cumarcs$ and substituting it in~\eqref{eq:matching cost on line}.

Now suppose that $\env$ is a circle instead of a line;
for example, the circle of unit radius $[0,2\pi)$
with distance between points $\coordvar_1$ and $\coordvar_2$ defined by
$\min\{ |\coordvar_1-\coordvar_2|, 2\pi - |\coordvar_1-\coordvar_2| \}$.
Though the circular case is quite similar to the linear one,
the identity permutation is not necessarily minimal for sorted $\pointsone$ and $\pointstwo$, because
matches across the $\coordvar = 0$ boundary are now possible.
It is easy to argue, however, that
the minimal matching is among the $\numpoints$ circular shift permutations, which
leads immediately to an $O(\numpoints^2)$ minimal matching algorithm, e.g., the one in~\cite{werman1985distance}.
The $O(\numpoints^2)$ barrier was indeed broken in~\cite{kuhn1955hungarian}, where
the authors observed that $\postcumarcs(\coordvar) = \cumarcs(\coordvar) + \numarcs$
for some integer $\numarcs$.
Since $\cumarcs(0)=0$, then
$\numarcs = \postcumarcs(0)$, which means
$\numarcs$ can be thought of as the (signed) number of matches crossing $y=0$.
Now the minimal matching has cost
\begin{equation} \label{eq:circle cost}
	\cost(\numarcs) =
	\int_0^{2\pi} | \cumarcs(y) + \numarcs | \ dy.
\end{equation}
%
%
%
Certaintly, the cost of the optimal matching can be no less than the minimum value of~\eqref{eq:circle cost}
taken over integer $\numarcs$.
\cite{werman1986bipartite} provides
an $O(\numpoints \log \numpoints)$ algorithm
to compute such minimum
and an $O(\numpoints)$ follow-up procedure
to produce a matching of cost no greater than $\cost(\numarcs)$.

\section{Problem Statement} \label{sec:problem statement}

\subsection{Notation}

We use the following graph notation throughout the paper:
Let $(\nodeset,\arcset)$ denote a directed graph, or \emph{di-graph},
with vertex set $\nodeset$ and a set of directed edges $\arcset$.
In general, $(\nodeset,\arcset)$ might be a \emph{multi-} di-graph,
meaning there may be multiple distinct edges having the same endpoints.
For any edge $\arcvar \in \arcset$,
we let $\arcvar^-$ denote the \emph{tail} of $\arcvar$ and
let $\arcvar^+$ denote the \emph{head} of $\arcvar$.
For example, if $\arcvar=(u,v)$, then $\arcvar^- = u$ and $\arcvar^+ = v$.

\begin{definition}[Orientation]
An \emph{orientation} of an undirected graph $G$ is the assignment of a direction to each edge in $G$,
resulting in a directed graph.
\end{definition}

%

\begin{definition}[Topological ordering]
A \emph{topological ordering}
of a directed acyclic graph $G=(\nodeset,\arcset)$ is
a linear ordering $\leq$ of the vertex set $\nodeset$, so that
for every $\arcvar \in \arcset$ it holds $\arcvar^- \leq \arcvar^+$.
\end{definition}
It is know that any directed acyclic graph (DAG) has at least one topological ordering.
An algorithm to compute a topological ordering of a DAG in time $O(|\nodeset| + |\arcset|)$
can be found in~\cite[Sec.~22.4]{leiserson2001introduction}.
%

\subsection{Roadmaps}

\providecommand{\roadnet}{{\mathcal R}}

\providecommand{\roadverts}{{\bf V}}
\providecommand{\roadset}{{\bf R}}
\providecommand{\roadvar}{r}
\providecommand{\roadlen}{\length}

A roadmap can be described in terms of a set of lines or curves
connected together into a particular pattern by their endpoints;
the distance between points on a roadmap is the minimum distance by which
a particle could reach one point from the other
while constrained to travel on the curves, or \emph{roads}.

It is common practice, e.g., by modern postal services,
to represent the topology of a roadmap using an undirected weighted graph or multi-graph $(\roadverts,\roadset)$,
possibly with loops,
where the edges $\roadset$ correspond to roads in the roadmap and are labeled with \emph{lengths}, and
the vertices $\roadverts$ describe their interconnections.
Another common practice is to attach to such graph a coordinate system:
Given a fixed \emph{orientation} of the roadmap graph,
every point on the roadmap continuum can be described unambiguously by a tuple, or \emph{address} $(\roadvar,y)$,
of a road $\roadvar \in \roadset$ and a real-valued coordinate $\coordvar$ between $0$ and the length $\roadlen_\roadvar$ of $\roadvar$.
%
There is an intuitive notion of ``roadmap distance'' between points described by such addresses,
arising from two basic assertions:
(i) there is a path between any two points on the same road, of length equal to
the difference between their address coordinates;
(ii) there is a special point for every roadmap vertex $u \in \roadverts$ which 
is at the respective endpoints of all the roads adjacent to $u$, simultaneously.
%
The distance between points then is the length of the shortest path between them.

\subsection{Objective}

The objective of the paper is to obtain an algorithm for the bipartite matching problem on roadmaps which
for a given roadmap $\roadnet$
has worst-case runtime bounded by $O(\numpoints \log \numpoints)$, where 
$\numpoints$ is the number of points in each of $\pointsone$ and $\pointstwo$.

\begin{lemma}[Lemma~2 of~\cite{werman1986bipartite}]
$\Omega(\numpoints \log \numpoints)$ is a lower bound for the time it takes to find the minimal cost matching in both the linear and the circular cases.
\end{lemma}
\begin{proof}[Proof Sketch]
The lemma can be proved by a reduction of the Set Equality problem, as in~\cite{werman1986bipartite}.
\end{proof}
$\Omega(\numpoints \log \numpoints)$ is clearly a lower bound
for the time to find a minimal cost matching on a fixed but arbitrary roadmap, because 
lines and circles are only two specific kinds of roadmaps.

%
%

\section{Optimal Bipartite Matching on a Roadmap} \label{sec:roadmap matching}

\providecommand{\surplus}{b}	
\providecommand{\costvar}{c}

\newcommand{\region}{{\mathcal Y}}
\newcommand{\costoffset}{\alpha}
\newcommand{\costslope}{\beta}

\providecommand{\lineset}{{\mathscr L}}

In this section, we generalize the analysis of~\cite{werman1986bipartite},
about the cost of matchings on lines and circles,
for the case of any fixed but arbitrarily complex roadmap $\roadnet$.
The cost of a match is equal to the 
\emph{roadmap distance} between the points, i.e.,
the length of the shortest path between them.
The result of our analysis will provide insight about the design of a novel minimal matching algorithm.

\subsection{Cost Characterization}

Since shortest paths are minimal, they are also simple (they do not ``cross themselves'').
We will attribute to every match $(\pointvar,\altpointvar)$ the orientation of its shortest path,
in the direction from $\pointvar$ to $\altpointvar$.
Note that a match need not have the same orientation on every road;
wherever its path follows a road $\roadvar$ in the positive direction,
we say it is forward;
wherever it follows a road in the opposite direction,
we say it is backward.
Without loss of generality, we assume that shortest paths are unique, e.g., by perturbation or arbitrary tie-breaking.

We will now re-define the quantities $\cumarcs$ and $\postcumarcs$ for roadmaps:
For every road $\roadvar \in \roadset$, $0 \leq \coordvar \leq \roadlen_\roadvar$,
let $\cumpoints_\pointsone(\coordvar;\roadvar)$ denote the number of points $(\roadvar,\coordvar_i)$
in $\pointsone$ which are on $\roadvar$
such that $\coordvar_i < \coordvar$, and
let $\cumpoints_\pointstwo(\coordvar;\roadvar)$ denote the number of points $(\roadvar,\coordvar_i)$
in $\pointstwo$ which are on $\roadvar$
such that $\coordvar_i < \coordvar$.
Let $\cumarcs(\coordvar;\roadvar) \doteq \cumpoints_\pointsone(\coordvar;\roadvar) - \cumpoints_\pointstwo(\coordvar;\roadvar)$.
Note that $\cumarcs(0;\roadvar) = 0$ for all $\roadvar \in \roadset$, but
now we have
$\cumarcs(\roadlen_\roadvar;\roadvar)
	= |\pointsone \cap \roadvar| - |\pointstwo \cap \roadvar|
	=: \surplus_\roadvar \neq 0$
in general.
Let $\surplus_\roadvar$ be called the \emph{surplus} of road $\roadvar$.

Given a minimal matching $\Match^*$,
let $\postcumarcs(\coordvar;\roadvar)$ denote
the number of forward matches crossing the coordinate $(\roadvar,\coordvar)$,
minus the number of backward matches crossing it.
%
Using the same arguments in~\cite{werman1986bipartite}, it is easy to argue that 
$\postcumarcs(\coordvar;\roadvar) = \cumarcs(\coordvar;\roadvar) + \numarcs_\roadvar$
for some integer $\numarcs_\roadvar$;
there is one such integer for every road $\roadvar \in \roadset$, and
we can write the total cost of all the match fragments using road $\roadvar$ as
%
\begin{equation} \label{eq:cost form}
\cost(\numarcs_\roadvar;\roadvar)
	= \int_0^{\roadlen_\roadvar}
		\left| \cumarcs(\coordvar;\roadvar) + \numarcs_\roadvar \right| \ d\coordvar.
\end{equation}

\begin{figure}[h!]
\centering
\input{generated/plotF.tex}
\caption{
An example road, with the points $\pointsone$ denoted by `$\times$' and the points $\pointstwo$ denoted by `$\circ$'.
}
\end{figure}
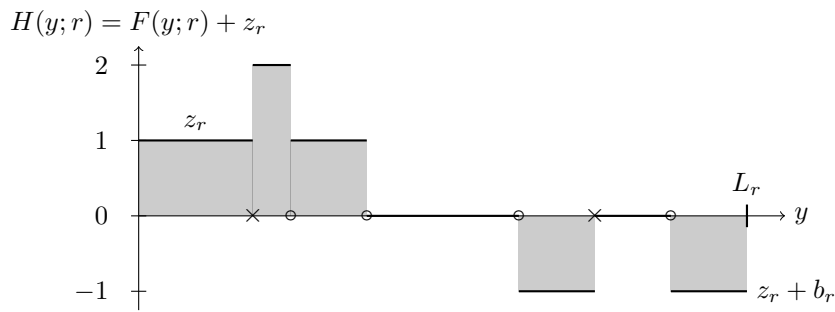

\begin{lemma}[Cost properties] \label{lemma:cost properties}
Letting $\cost$ be defined as in~\eqref{eq:cost form}
for each $\roadvar \in \roadset$ and
all $\numarcs \in \reals$ (and not just the integers),
$\cost$ is
(i) piece-wise linear,
(ii) convex, and
(iii) unbounded, and
(iv) has constant slope in every integer interval.
\end{lemma}
\begin{proof}
We prove the result for a single road $\roadvar$, omitting the $\roadvar$-specific notation.
Note that $\cumarcs$ is piece-wise constant and takes value only on the set of integers.
Let $\cumvals = \{ \cumval_1, \ldots, \cumval_m \} \subset \integers$ denote
the set of values taken by $\cumarcs$ (in increasing order), and
let $I_1,\ldots,I_{m+1}$ denote the intervals
$(-\infty,\cumval_1),(\cumval_1,\cumval_2),\ldots,$
	$(\cumval_{m-1},\cumval_m),(\cumval_m,+\infty)$,
disjointly covering $\reals$.
%
$\cost(\numarcs)$ can be written as
\begin{equation} \label{eq:cost form no abs}
\cost(\numarcs) =
	\int_{ \cumarcs(y) + \numarcs > 0 } [ \cumarcs(y) + \numarcs ] \ dy
	- 	\int_{ \cumarcs(y) + \numarcs < 0 } [ \cumarcs(y) + \numarcs ] \ dy.
\end{equation}
For each $k=1,\ldots,m+1$ and all $\numarcs \in -I_k$ (i.e., $-\numarcs \in I_k$),
we have that
$\{ y : \ \cumarcs(y) + \numarcs > 0 \} \equiv \{ \cumarcs(y) > \cumval_{k-1} \} =: \region^+_k$ and
$\{ y : \ \cumarcs(y) + \numarcs < 0 \} \equiv \{ \cumarcs(y) \leq \cumval_{k-1} \} =: \region^-_k$ are both constant.
(Here, we let $\cumval_0 := -\infty$.)
Defining scalar constants
\[
	\costoffset_k \doteq \int_{ \region^+_k } \cumarcs(y) \ dy - \int_{ \region^-_k } \cumarcs(y) \ dy,
\]
and
\begin{equation}
\label{eq:cost piece slope}
	\costslope_k \doteq |\region^+_k| - |\region^-_k|,
\end{equation}
then for all $\numarcs \in -I_k$,
\eqref{eq:cost form no abs} can be written as $\cost(\numarcs) = \costoffset_k + \numarcs \costslope_k$.

The prequel demonstrates that $\cost$ is piece-wise linear 
with constant slope within every integer interval.
To show that $\cost$ is convex, one can show that the slope of $\cost$ is non-decreasing,
e.g., by confirming that $\costslope_k$ is non-increasing in $k$.
Unboundedness can be proved by confirming that $\costslope_1 = \roadlen$ and $\costslope_{m+1} = -\roadlen$, where
it can be assumed that the road length $\roadlen > 0$.

%
%
\end{proof}

\subsection{Cost Bounds on Optimal Matchings}
\label{sec:reduction}

The next problem can be used to bound the cost of the minimal matching from below, and
we will eventually demonstrate that its solution can be used to produce a matching which achieves the bound.
\begin{myprob}[Minimum cost of matching]
\label{prob:matching cost lower bound}
\begin{align}
\underset{
	\Numarcs \in \reals^\roadset
}{\text{minimize}}
	&&& \cost(\Numarcs) = \sum_{\roadvar \in \roadset} \cost( \numarcs_\roadvar ; \roadvar )		
	\label{eq:cost of matching} \\
	%
\text{subject to}
	&&& \sum_\SetOf{ \roadvar \SuchThat \roadvar^+ = u } \numarcs_\roadvar + \surplus_\roadvar
	= \sum_\SetOf{ \roadvar \SuchThat \roadvar^- = u } \numarcs_\roadvar
	\qquad \text{for all vertices $u$}.
	\label{eq:conservation}
\end{align}
\end{myprob}

\begin{lemma}
The minimal matching has cost bounded below by~\eqref{eq:cost of matching}
of the optimal solution of Problem~\ref{prob:matching cost lower bound}.
\end{lemma}
\begin{proof}
Given a matching $\Match$,
one can compute
$\numarcs_\roadvar := \postcumarcs( \argholder ; \roadvar ) - \cumarcs( \argholder ; \roadvar )$
for each $\roadvar \in \roadset$.
If $\Numarcs \in \integers^\roadset$ is the vector composed of such $\numarcs_\roadvar$, then
the total cost of $\Match^*$ is precisely~\eqref{eq:cost of matching}.
%
For each $u \in \roadverts$,
the left hand side of~\eqref{eq:conservation} is equal to the number of matches entering $u$, and
the right hand side is equal to the number of matches leaving $u$ (both \emph{signed} counts).
Regular conservation arguments imply that~\eqref{eq:conservation} must hold for any feasible matching.
%
Therefore, the feasible set of Problem~\ref{prob:matching cost lower bound} contains all realizable $\Numarcs$, obtaining the lemma.
\end{proof}

The previous formulation can be extended quite easily to incorporate the notion of ``one-way'' roads:
For example, if a road $\roadvar \in \roadset$ admits only forward matches,
that condition can be encoded as $\postcumarcs( \coordvar ; \roadvar ) \geq 0$
for all $\coordvar \in (0,\roadlen_\roadvar)$, i.e.,
$\min_\coordvar \cumarcs(\coordvar;\roadvar) + \numarcs_\roadvar \geq 0$.
%
%
Note therefore that the framework can support one-way and bi-directional roads simultaneously.

\begin{lemma}[Integral solutions] \label{lemma:integral solutions}
The optimization problem Problem~\ref{prob:matching cost lower bound} has integer optimal solutions.
\end{lemma}
The significance of Lemma~\ref{lemma:integral solutions} is that we may hope to avoid the complexity
associated with integer optimization problems
by linear relaxation.
\begin{proof}
Suppose that $\Numarcs$ is a fractional optimal solution to~\eqref{prob:matching cost lower bound}, and
let $\roadvar_1$ be a road with non-integer component, i.e., $\numarcs_{\roadvar_1} \notin \integers$.
Note that since all the coefficients of~\eqref{eq:conservation} are integer,
then among the roads which share endpoint $\roadvar_1^+$ with $\roadvar_1$
(including itself if $\roadvar_1^- = \roadvar_1^+$),
at least one must also be fractional.
Using this argument, one can identify a cycle $(\roadvar_1,\ldots,\roadvar_K)$ of \emph{only} roads with fractional components.
For each $k=1,\ldots,K$, let us denote
by $I_k \subset \reals$ the unit interval $(\lfloor \numarcs_{\roadvar_k} \rfloor, \lceil \numarcs_{\roadvar_k} \rceil)$,
and let $\Pi$ denote the unit hyper-cube $\prod_k I_k$.
Note that $(\numarcs_{\roadvar_1},\ldots,\numarcs_{\roadvar_K})$ is in the \emph{interior} of $\Pi$.
In the interior of $\Pi$, we have a gradient
\[
	\gradient_{ (\numarcs_{\roadvar_1},\ldots,\numarcs_{\roadvar_K}) } \cost(\Numarcs)
	= \begin{bmatrix}
	\left.
         \frac{d}{d\numarcs}
         \cost( \numarcs; \roadvar_1 )
      \right|_{\numarcs_{\roadvar_1}}
      &
   \ldots
      &
	\left.
         \frac{d}{d\numarcs}
         \cost( \numarcs; \roadvar_K )
      \right|_{\numarcs_{\roadvar_K}}
	\end{bmatrix}.
\]

Applying Lemma~\ref{lemma:cost properties}, one can argue that such gradient is
constant over $\Pi$.
In particular, the gradient must be $0$, since $\Numarcs$ is optimal.
Let $g \in \reals^{\roadset}$ denote the vector which is
$+1$ for each road traversed in the positive direction by the cycle,
$-1$ for each road traversed in the reverse direct,
and $0$ for all roads not in the cycle.
(It is easy to show that $\Numarcs + \alpha g$ satisfies~\eqref{eq:conservation} for all $\alpha \in \reals$.)
One can find $\alpha \in \reals$ such that $\Numarcs + \alpha g =: \Numarcs^+$ lies on the boundary of $\Pi$.
$\Numarcs^+$ is also feasible optimal, and has at least one less fractional component than $\Numarcs$.
Such procedure can be repeated until an integral optimal solution is obtained.
\end{proof}

\subsubsection{Solving Problem~\ref{prob:matching cost lower bound}}
\label{sec:solving convex cost flow problem}

Since by Lemma~\ref{lemma:cost properties} the objective functions $\cost(\argholder;\roadvar)$ are all convex,
Problem~\ref{prob:matching cost lower bound} is a so-called [minimum] \emph{convex cost flow problem}~\cite[Ch.~14]{ahuja1993network}.
The convex cost flow problem (CCFP) is a generalization of the minimum [linear] cost flow problem, such that
the edge costs needn't be linear.
Ahuja et. al. provide a \emph{capacity-scaling} algorithm
for CCFPs like Problem~\ref{prob:matching cost lower bound} which have integral solutions~\cite[Sec.~14.5]{ahuja1993network}.
Provided that the objective functions can be evaluated in $O(1)$ time,
the algorithm obtains runtime $O( ( \numedges \log U ) S(\numnodes,\numedges,C) )$, where
$\numnodes$ and $\numedges$ are the number of vertices and edges in the network, respectively,
$U$ and $C$ are bounds on the total supply and cost, respectively, and
$S(\numnodes,\numedges,C)$ is the time to solve the shortest path problem on such a network~\cite[Ch.~4]{ahuja1993network}.
Choosing, for example, the strongly polynomial $O(\numedges + \numnodes \log \numnodes)$
Fibonacci heap shortest-path algorithm of Fredman and Tarjan~\cite{fredman1987fibonacci}, and
observing that $U = \numpoints$ is an acceptable bound,
one can solve Problem~\ref{prob:matching cost lower bound} in
$O( \numedges \log M ( \numedges + \numnodes \log \numnodes ) )$ time.
Moreover, the algorithm always provides integer solutions.

\subsection{Obtaining an Optimal Roadmap Matching}
\label{sec:construction}

\providecommand{\edgelen}{l}
\providecommand{\postcumval}{h}

In this section, we provide an algorithm for the construction of an optimal matching given the vector $\Numarcs^*$.
The algorithm generalizes, in a fairly clean way, e.g., the procedure given in~\cite{werman1986bipartite}.


Given a solution $\Numarcs$ to Problem~\ref{prob:matching cost lower bound},
the following procedure will obtain a matching with cost less than or equal to $\cost(\Numarcs)$:
The procedure is in two steps.
During the first step, one obtains an intermediate data structure, of
an \emph{integer}-weighted digraph $G(\Numarcs)$
on vertex set $\pointsone \cup \pointstwo \cup \roadverts$.
In such graph, every edge corresponds to an empty interval on the roadmap
(generally, one spanning the space between two adjacent points);
thus, we will call $G(\Numarcs)$ the \emph{interval graph}.
In the second step, to obtain the matching itself,
we run a simple graph traversal algorithm (Algorithm~\ref{alg:construct}) on $G(\Numarcs)$.

The edges of the interval graph are directed, and each is labeled with a positive integer ${weight}$,
indicating the number of matches which cross the corresponding interval.
For example, let $(\coordvar_1,\coordvar_2,\ldots,\coordvar_K)$ be the ordering
of the coordinates of points in $\pointsone \cup \pointstwo$ which are on some road $\roadvar$.
Let $Y_1,Y_2,\ldots,Y_{K+1}$ denote the set of intervals
$(0,\coordvar_1)$,$(\coordvar_1,\coordvar_2)$,$\ldots$,
$(\coordvar_{K-1},\coordvar_K)$,$(\coordvar_K,\roadlen_\roadvar)$, and
let $\cumval_k$ denote the constant value taken by $\cumarcs(\cdot;\roadvar)$ per interval $Y_k$.
(Let $\edgelen_k$ denote the interval length.)
Note that $\cumval_k + \numarcs_\roadvar =: \postcumval_k$ is the signed number of matches traversing $Y_k$ under $\Numarcs$.
If $\postcumval_k > 0$ then $G(\Numarcs)$ contains the edge
from the earlier endpoint $(\roadvar,\coordvar_{k-1})$ to the later endpoint $(\roadvar,\coordvar_k)$;
if $\postcumval_k < 0$, then it has the reverse edge;
if $\postcumval_k = 0$, then there is no edge.
If the edge is contained, then it has ${weight}$ $|\postcumval_k|$.
Note that the endpoints $(\roadvar,0)$ and $(\roadvar,\roadlen_\roadvar)$,
of intervals $Y_1$ and $Y_K$, respectively,
are not points in $\pointsone \cup \pointstwo$.
They are substituted in $G(\Numarcs)$ by the roadmap vertices $\roadvar^- \in \roadverts$ and $\roadvar^+ \in \roadverts$, respectively,
to which they correspond.
\begin{remark}
\label{remark:graphweight}
$\cost(\Numarcs)$ is equal to
the sum over all edges in $G(\Numarcs)$ of $|\postcumval| \times \edgelen$.
\end{remark}

Algorithm~\ref{alg:construct} (below) produces a matching by visiting every vertex in $G(\Numarcs)$,
among which are all of $\pointsone \cup \pointstwo$.
When the current vertex $i$ is a point in $\pointsone$, then it is ``collected'' for future use.
When $i \in \pointstwo$, then it is matched to a point $j$ previously collected.
To ensure that there are always enough points collected for future matches, $G(\Numarcs)$ is traversed in a topological order.
%
\providecommand{\toposeq}{\sigma}
\begin{algorithm}[h!]
\caption{ConstructMatching}
\label{alg:construct}
\begin{algorithmic}[1]
\renewcommand{\algorithmicrequire}{\textbf{Input:}}
\renewcommand{\algorithmicensure}{\textbf{Output:}}
\renewcommand{\algorithmiccomment}[1]{ // \emph{#1} }
\REQUIRE an interval di-graph $G$
(e.g., generated by inputs $\roadnet$, $\pointsone$, $\pointstwo$, and $\Numarcs$)
\ENSURE a bipartite matching $\Match$ between $\pointsone$ and $\pointstwo$

\ensurecommand{\currnode}{v}
\ensurecommand{\somenode}{u}
\ensurecommand{\nextnode}{w}

\STATE {\bf initialize}:
$\Match \gets$ an empty matching

\STATE {\bf initialize}:
Associate
with each vertex $\currnode \in \pointsone \cup \pointstwo \cup \roadverts$
an empty set $L_\currnode$.

\STATE Choose any topological ordering
of $\pointsone \cup \pointstwo \cup \roadverts$
under $G(\Numarcs)$.
Enumerate $\pointsone \cup \pointstwo \cup \roadverts$ in this order:

\FORALL{vertices $\currnode$}
\STATE \label{line:programstart}
If $\currnode \in \pointsone$, then add $\currnode$ to $L_\currnode$.
Otherwise, if $\currnode \in \pointstwo$ then remove some point $\somenode$ from $L_\currnode$, and
insert the match $(\somenode,\currnode)$ into $\Match$.
(Because the vertices are enumerated in a topological order,
if $\currnode \in \pointstwo$ then it must be true $|L_\currnode| > 0$.)
If $\currnode \in \roadverts$, we do not alter $L_\currnode$, nor create a match. 

\STATE \label{line:programend}
For every edge $(\currnode,\nextnode)$ leaving $\currnode$,
move $weight(\currnode,\nextnode)$ elements from $L_\currnode$ into $L_\nextnode$.
\ENDFOR
\RETURN $\Match$
\end{algorithmic}
\end{algorithm}

\begin{lemma}
The cost of the matching produced by Algorithm~\ref{alg:construct} on interval graph $G(\Numarcs)$ has cost no greater than $\cost(\Numarcs)$.
\end{lemma}
\begin{proof}
During the execution of Algorithm~\ref{alg:construct},
every point $\pointvar \in \pointsone$
traverses a path in $G(\Numarcs)$ to its match $\Match(\pointvar) \in \pointstwo$.
The cost of the matching produced cannot be greater than the total length of all such paths.
By design, every edge in $G(\Numarcs)$ is traversed
under Algorithm~\ref{alg:construct}
by a total number of paths
equal to its weight.
The total length of all paths is equal therefore to the sum over the edges in $G(\Numarcs)$
of their weight times length.
That sum is equal to $\cost(\Numarcs)$; recall, e.g., Remark~\ref{remark:graphweight}.
\end{proof}
Algorithm~\ref{alg:construct} has one technical caveat:
It assumes that a topological ordering exists under $G(\Numarcs)$.
This is guaranteed as long as $G(\Numarcs)$ is a DAG, since
every DAG has at least one topological ordering.
\begin{lemma}
Let $\Numarcs^*$ be an \emph{optimal} solution to Problem~\ref{prob:matching cost lower bound}.
$G(\Numarcs^*)$ is a DAG.
\end{lemma}
\begin{proof}
The proof is by contradiction.
Suppose $\Numarcs^*$ is optimal, but $G(\Numarcs^*)$ has a directed cycle
$\cyclevar = (\arcvar_1,\ldots,\arcvar_K)$.
$\cyclevar$ corresponds to a cycle
$\cyclevar' = (\roadvar_1,\ldots,\roadvar_{K'})$ of the roads in $\roadset$.
Let $g \in \reals^\roadset$ denote the vector which is
$+1$ for each road traversed in the positive direction by $\cyclevar'$,
$-1$ for each road traversed in the reverse direction,
and $0$ for all roads not in $\cyclevar'$.
$\Numarcs^*-g$ is feasible because it satisfies~\eqref{eq:conservation}.
It is easy to argue that we can obtain $G(\Numarcs^*-g)$ from $G(\Numarcs^*)$
by substracting $1$ from the weight on every edge in $\cyclevar$.
Doing so strictly decreases the total weight of the edges in $G$,
ultimately decreasing $\cost$, and thereby
contradicting the optimality of $\Numarcs^*$.
\end{proof}

In fact, $G(\Numarcs^*)$ is a special kind of DAG:
\begin{prop}
$G(\Numarcs^*)$ is a \emph{multi-tree}, i.e., a directed, acyclic graph in which there is at most \emph{one} directed path between any two vertices.
\end{prop}
\begin{proof}
The proof is by contradiction, and is similar to the previous one.
Suppose that $G(\Numarcs^*)$ is optimal, but there are two distinct paths, $\pathvar$ and $\pathvar'$,
from some vertex $u$ to a vertex $v$.
(Recalling our assumption (w.l.g.) that shortest paths are unique, suppose $\pathvar$ is strictly shorter than $\pathvar'$.)
Let $\cyclevar$ denote a cycle which
traverses $\pathvar$ in the forward direction, and
traverses $\pathvar'$ in the reverse direction.
Let $\cyclevar'$ denote the corresponding cycle on roads in $\roadset$.
Defining $g$ as before, $\Numarcs^* + g$ is feasible, and
we can obtain $G(\Numarcs^* + g)$ from $G(\Numarcs^*)$ by incrementing the edge weights on $\pathvar$ and
decrementing the edge weights on $\pathvar'$.
(We remove any edges which obtain zero weight.)
This action decreases $\cost$, contradicting the optimality of $\Numarcs^*$.
\end{proof}

\subsection{Complexity Analysis} \label{sec:complexity}

The analysis of the prequel suggests a minimal matching algorithm in three fundamental steps.
\begin{enumerate}
\item \label{item:transcribe}
Transcribe the instance of Problem~\ref{prob:matching cost lower bound}
generated by $\roadnet$, $\pointsone$, and $\pointstwo$;

\item \label{item:solve}
Obtain an optimal integer solution $\Numarcs^*$;

\item \label{item:construct}
Use the solution vector $\Numarcs^*$ to construct a matching $\Match^*$
with cost $\cost(\Match^*) = \cost(\Numarcs^*)$, following the procedure of Section~\ref{sec:construction}.
\end{enumerate}
In this section we will demonstrate that all three steps can be completed
within $O(\numpoints \log \numpoints)$ time, in terms of $\numpoints$.

\subsubsection{Transcribing Problem~\ref{prob:matching cost lower bound}} \label{sec:transcribing}

\providecommand{\intervalmap}{{\mathcal Y}}
\providecommand{\cumvalarr}{{\mathcal F}}
\providecommand{\cumvalmap}{{\mathcal Z}}
\providecommand{\levelmap}{{\mathcal I}}
\providecommand{\linesmap}{{\mathcal C}}

\providecommand{\numroadpoints}{{K_1}}
\providecommand{\numlevels}{{K_2}}

An instance of Problem~\ref{prob:matching cost lower bound} is specified by
the vertex supplies $\surplus$ and the cost functions $\cost$, which we must therefore compute.
We rely on the key fact of Prop.~\ref{prop:number of lines upper bound} below.
\begin{prop} \label{prop:number of lines upper bound}
There are at most $2\numedges + 2\numpoints$ total linear pieces among the objectives $\cost(\argholder;\roadvar)$.
\end{prop}
\begin{proof}
If a road $\roadvar \in \roadset$ has no points,
then $\cumarcs(\cdot;\roadvar)$ equals $0$ everywhere on $\roadvar$, i.e.,
$\cumarcs(\cdot;\roadvar)$ takes values from the set $\{ 0 \}$, and so
$\cost(\cdot;\roadvar)$ is linear over each of the intervals $(-\infty,0)$ and $(0,+\infty)$.
The $\numedges = |\roadset|$ total roads result in $2\numedges$ such ``free'' intervals.
Any point added to a road $\roadvar$ may introduce at most one additional value in the value set of $\cumarcs(\cdot;\roadvar)$.
If it does, then one of the pieces of $\cost(\cdot;\roadvar)$ is split into two.
There are at most $2\numpoints$ points, so the lemma holds.
The bound can be realized as long as $|\roadset| \geq 2$,
by placing all of $\pointsone$ on one road and
all of $\pointstwo$ on another.
\end{proof}

One way to obtain the data is by computing the following collection of arrays,
one per road $\roadvar \in \roadset$:
\begin{description}
\item[$\intervalmap_\roadvar = 
      \begin{bmatrix} y_0 & \ldots & y_{\numroadpoints+1} \end{bmatrix}$] 
a sorted array of
the $\numroadpoints$ locations of points on $\roadvar$; notably
$\coordvar_0 = 0$, and $\coordvar_\numroadpoints = \roadlen_\roadvar$

\item[$\cumvalarr_\roadvar = 
         \begin{bmatrix} f_0 & \ldots & f_{\numroadpoints} \end{bmatrix}$]
an array of $ f_i := \cumarcs(y_i;r)$ for each $y_i \in \intervalmap_\roadvar$; note,
$\cumarcs(y;\roadvar) = \cumval_i$ for all $y \in [y_i, y_{i+1})$

\item[$\cumvalmap_\roadvar = 
         \begin{bmatrix} z_0 & \ldots & z_{\numlevels} \end{bmatrix}$]
a sorted array of
[$z_0 = -\infty$, and]
the $\numlevels$ consecutive integers 
appearing in $-\cumvalarr_\roadvar$; notably,
for all $i \in \SetOf{1 \ldots \numlevels}$ and for some $j$, $z_i = -f_j$ 

\item[$\levelmap_\roadvar = 
         \begin{bmatrix} I_0 & \ldots & I_\numlevels \end{bmatrix}$]
an array of ``level measures'' such that 
$I_i = \sum_J y_{j+1} - y_{j}$ where $J = \SetOf{j \SuchThat -f_j = z_i}$

\item[$\linesmap_\roadvar = 
         \begin{bmatrix} c_0 & \ldots & c_{\numlevels} \end{bmatrix}$]
an array of pairs $c_k = (\costoffset_k,\costslope_k)$ such that 
$\cost(z;\roadvar) = \costoffset_k + \costslope_k z$, for $z \in [z_{i}, z_{i+1}]$. 
\end{description}
For every road $\roadvar \in \roadset$,
$\surplus_\roadvar$ is equal to the last entry of $\cumvalarr_\roadvar$, and
$\linesmap_\roadvar$ contains the data of the linear pieces of $\cost(\argholder;\roadvar)$ in order of appearance.

Since the road set $\roadset$ is known ahead of time,
each point in $\pointsone \cup \pointstwo$ can be associated
to a collection $\intervalmap_\roadvar$ in $O(1)$ time (using, e.g., a hash function).
Therefore, sorting $\pointsone \cup \pointstwo$ into the collections $\{ \intervalmap_\roadvar \}$
can be done in a single sweep of $\pointsone \cup \pointstwo$ in $O(\numpoints \log \numpoints)$ time.
Each $\cumvalarr_\roadvar$ can be generated by a linear scan (accumulation) over $\intervalmap_\roadvar$, and so
the entire collection $\{ \cumvalarr_\roadvar \}$ can be obtained in $O( \numedges + \numpoints )$ time;
the $O(\numedges)$ term appears if every road is processed explicitly, even if it contains no points.
%

$\cumvalmap_\roadvar$ can be built by sorting the distinct values appearing in $\cumvalarr_\roadvar$, and so 
it takes at most $O( \numedges + \numpoints \log \numpoints)$ time to obtain the whole collection.
$\levelmap_\roadvar$ can be built by the following procedure: 
Initialize the array to all zeros.
Perform a scan over $\cumvalarr_\roadvar$.
For each $k=0,\ldots,\numlevels$, add
the interval length $\coordvar_{k+1} - \coordvar_{k}$
to the level measure $I_j$ such that $\numarcs_j = -f_k$.
There are $O(\numpoints)$ elements among all the $\cumvalarr_\roadvar$, so
the total time to build the collection $\{ \levelmap_\roadvar \}$ is
$O( \numedges + \numpoints )$.
Finally each $\linesmap_\roadvar$ can be built by a scan
over
$\levelmap_\roadvar$, e.g.,
using~\eqref{eq:cost piece slope},
and the whole collection can be obtained in $O(\numedges + \numpoints)$ time . 

Adding up the times of all the steps, we obtain 
$O( \numedges + \numpoints \log \numpoints)$ time to compute all arrays.
(The dependence on $\numedges$ can actually be removed, because
all the arrays have implicit value for any road which obtains no points.)

%

\subsubsection{Solving Problem~\ref{prob:matching cost lower bound} for $\Numarcs^*$}

As demonstrated in Section~\ref{sec:solving convex cost flow problem}, $\Numarcs^*$ can be obtained in
$O( \numedges ( \numedges + \numnodes \log \numnodes ) \log \numpoints  )$ time
by a capacity-scaling algorithm,
provided an $O(1)$ procedure to compute each $\cost(\argholder;\roadvar)$.
One such procedure is
to find the linear data $(\costoffset_k,\costslope_k)$
associated with a query point $\numarcs$, and then
compute $\cost(\numarcs;\roadvar) = \costoffset_k + \costslope_k \numarcs$.
%
The first task can be accomplished by simple positional lookup in the array $\linesmap_\roadvar$, since
the domain of the array is an ordered list of consecutive integers (i.e., $\cumvalmap_\roadvar$, except $-\infty$).

\subsubsection{Obtaining a Minimal Matching Given $\Numarcs^*$}

To obtain the final matching one must:
(i) construct the graph $G(\Numarcs^*)$,
(ii) compute a topological ordering of the vertices under $G(\Numarcs^*)$, and
(iii) execute a short program (lines~\ref{line:programstart}-\ref{line:programend} of Algorithm~\ref{alg:construct}) once per vertex.
The graph $G(\Numarcs^*)$ can be obtained in $O(\numedges + \numpoints)$ time
by enumerating the intervals in the arrays $\{ \cumvalarr_\roadvar \}$.
A topological ordering of the vertices
can be obtained in $O(\numnodes + \numpoints)$ time using a standard algorithm, e.g.,
the one in~\cite[Sec.~22.4]{leiserson2001introduction}.
If the sets in the algorithm are implemented using linked-list queues, then
the duties of the program per vertex are all constant time except possibly the task of splitting a list.
It can be checked however, that at each of the vertices in $\pointsone \cup \pointstwo$, the split is trivial and we can skip it.
In the worst case we may suffer at most $\numnodes = |\roadverts|$ non-trivial splits,
of at most $O(\numpoints)$ operations each.
Therefore, the total time of the construction is $O( \numedges + \numnodes \numpoints )$.
Using a slightly more sophisticated data structure can reduce this time to
$O( \numnodes\numedges + \numpoints )$:
Each sequence of adjacent points
stored in a linked-list queue
can be represented instead using a single index range.
Because $G(\Numarcs^*)$ is a multi-tree,
it can be proved that such queues will never obtain more than $2\numedges$ such ranges under Algorithm~\ref{alg:construct}.

\subsubsection{Total Runtime Complexity}

Adding the runtimes of all three components and re-arranging terms,
we find that the entire algorithm can be computed in
$O( \numpoints \log \numpoints ) +
	\log\numpoints \times
	O( \numedges ( \numedges + \numnodes \log \numnodes ) )$
time.
For $\numpoints$ sufficiently large compared to the description of the roadmap, the algorithm is dominated by
the act of transcribing the matching instance as an instance of the convex cost flow problem,
followed by construction of the matching given the optimal solution $\Numarcs^*$;
in comparison to those steps, the search for the optimal solution $\Numarcs^*$ is relatively easy.
Moreover, the space requirements of the algorithm are linear in $\numnodes$, $\numedges$, and $\numpoints$, and so
the algorithm is strongly polynomial.

\section{Discussion} \label{sec:discussion}

We would like to remark that
the generalizations in this paper of the techniques in~\cite{werman1986bipartite} are essentially straightforward.
Instead, the crucial contribution of the present result is recognizing that
Problem~\ref{prob:matching cost lower bound} can be solved efficiently and
by a strongly polynomial algorithm.
Such knowledge seems to be fairly weakly disseminated; for example, as far as we know,
Ahuja et. al. only published the relevant min-cost flow algorithm in their text book~\cite{ahuja1993network}.
Depending on the sophistication of one's standard approaches,
Problem~\ref{prob:matching cost lower bound} could alternatively be reduced to
a network flow problem with $O(\numpoints)$ edges, or to
a generic linear optimization problem in $O(\numpoints)$ constraints.
In either case, most standard contemporary algorithms fail to meet the quite demanding $O(\numpoints \log \numpoints)$ runtime budget;
moreover, strongly polynomial algorithms tend to be elusive.
While the strongly polynomial property of our algorithm stems from the fact that the solution is integral (an assignment),
it might be possible to obtain a strongly polynomial, efficient algorithm for the more general \emph{transportation problem} on roadmaps.
This study highlights the fact that netflow flow problems are somewhat better understood,
in practice,
than general linear optimization problems.

\section*{Acknowledgements}

We thank Dr. Michael Otte and Kevin Spieser for helpful discussions.

\bibliographystyle{unsrt}
\bibliography{main}

\appendix

\section{Capacity Scaling Algorithm}

\providecommand{\flowvar}{x}
\providecommand{\potvar}{\pi}
\providecommand{\scalevar}{\Delta}
\providecommand{\capacity}{u}
\providecommand{\imbalance}{e}
\providecommand{\costvar}{c}

\begin{algorithm}[h!]
\caption{Capacity Scaling}
\label{alg:capacity scaling}
\begin{algorithmic}[1]
\renewcommand{\algorithmicrequire}{\textbf{Input:}}
\renewcommand{\algorithmicensure}{\textbf{Output:}}
\renewcommand{\algorithmiccomment}[1]{ // \emph{#1} }
\providecommand{\doInit}{{\bf initialize}:\ }
\REQUIRE
\ENSURE
\STATE \doInit flows $\flowvar := 0$ and potentials $\potvar := 0$
\STATE \doInit Compute imbalances $\imbalance$
\STATE \doInit $\scalevar := 2^{\lfloor \log U \rfloor}$
\WHILE{$\scalevar \geq 1$}
	\STATE [Re-] Compute $\scalevar$-residual graph $G(\flowvar;\scalevar)$ and reduced costs $\costvar^\potvar$
	\FOR{every arc $\arcvar$ in the residual network $G(x;\scalevar)$}
	\IF{$\costvar^\potvar_{\arcvar} < 0$}
		\STATE Send $\scalevar$ flow along arc $\arcvar$: update
			$\flowvar$, $\imbalance$, $G(\flowvar;\scalevar)$ and $\costvar^\potvar$
	\ENDIF
	\ENDFOR
	\STATE $\pointsone(\scalevar) := \{ i : \imbalance(i) \geq \scalevar \}$
	\STATE $\pointstwo(\scalevar) := \{ i : \imbalance(i) \leq -\scalevar \}$
	\WHILE{ $\pointsone(\scalevar) \neq \emptyset$ and $\pointstwo(\scalevar) \neq \emptyset$ }
		\STATE Choose $s \in \pointsone(\scalevar)$ and $t \in \pointstwo(\scalevar)$
		\STATE Determine shortest path distances $d(\cdot)$ from $s$ to all other nodes in $G(\flowvar;\scalevar)$
			with respect to the reduced costs $\costvar^\potvar$
		\STATE Let $\pathvar$ be the shortest path from $s$ to $t$ in $G(\flowvar;\scalevar)$
		\STATE Augment $\scalevar$ flow along the path $\pathvar$: update
				$\flowvar$, $\imbalance$, $G(\flowvar;\scalevar)$, $\pointsone(\scalevar)$, and $\pointstwo(\scalevar)$
		\STATE Update $\potvar := \potvar - d$ and reduced costs $\costvar^\potvar$
	\ENDWHILE	
	\STATE $\scalevar := \scalevar / 2$
\ENDWHILE
\RETURN $\flowvar$
\end{algorithmic}
\end{algorithm}

\end{document}

%% file: macros.tex
\usepackage{makecmds}

\newcommand{\emath}[1]{\ensuremath{#1}}

\ifdefined\kylehackusingbeamer
\else
	
	\newtheorem{theorem}{Theorem}[section]
	\newtheorem{definition}[theorem]{Definition}
	
	\newtheorem{lemma}[theorem]{Lemma}
	
	\newtheorem{remark}[theorem]{Remark}
\fi
	\newtheorem{prop}[theorem]{Proposition}

\provideenvironment{example}[1][Example]{
\begin{trivlist}
\item[\hskip \labelsep {\bfseries #1}]
\it	
}{ \end{trivlist} }

\newcommand{\reals}{ \emath{\mathbb{R}} }

\newcommand{\Transposed}{^\top}
\newcommand{\SetOf}[1]{{ \left\{ #1 \right\}  }}
\newcommand{\SuchThat}{\,|\,}

%% file: macros_specific.tex
\newcommand{\env}{\emath{\Omega}}

\newcommand{\cyclevar}{\mathcal L}

\newcommand{\permvar}{\emath{\sigma}}

\newcommand{\pathvar}{{\mathcal P}}

\newcommand{\nodeset}{V}

\newcommand{\edgeset}{{\mathcal A}}



%% file: generated/plotF.tex
\begin{tikzpicture}
\draw [->] (-4.000000,0) -- (4.500000,0) node [right] {$\coordvar$} ;
\draw [thick] (4.000000,-.15) -- (4.000000,.15) node [above] {$\roadlen_\roadvar$} ;
\draw [->] (-4.000000,-1.250000) -- (-4.000000,2.250000) node [above] {$\postcumarcs(y;\roadvar) = \cumarcs(y;\roadvar) + \numarcs_\roadvar$} ;
\draw (-3.900000,-1) -- (-4.100000,-1) node [left=5] {$-1$} ;
\draw (-3.900000,0) -- (-4.100000,0) node [left=5] {$0$} ;
\draw (-3.900000,1) -- (-4.100000,1) node [left=5] {$1$} ;
\draw (-3.900000,2) -- (-4.100000,2) node [left=5] {$2$} ;
\draw (-2.500000,0) node {$\times$} ;
\draw (-2.000000,0) node {$\circ$} ;
\draw (-1.000000,0) node {$\circ$} ;
\draw (1.000000,0) node {$\circ$} ;
\draw (2.000000,0) node {$\times$} ;
\draw (3.000000,0) node {$\circ$} ;
\draw [thick] (-4.000000,1.000000) -- node [above] {$\numarcs_\roadvar$} (-2.500000,1.000000) ;
\path [fill=black,opacity=.2] (-4.000000,0) rectangle (-2.500000,1.000000) ;
\draw [thick] (-2.500000,2.000000) --  (-2.000000,2.000000) ;
\path [fill=black,opacity=.2] (-2.500000,0) rectangle (-2.000000,2.000000) ;
\draw [thick] (-2.000000,1.000000) --  (-1.000000,1.000000) ;
\path [fill=black,opacity=.2] (-2.000000,0) rectangle (-1.000000,1.000000) ;
\draw [thick] (-1.000000,0.000000) --  (1.000000,0.000000) ;
\path [fill=black,opacity=.2] (-1.000000,0) rectangle (1.000000,0.000000) ;
\draw [thick] (1.000000,-1.000000) --  (2.000000,-1.000000) ;
\path [fill=black,opacity=.2] (1.000000,0) rectangle (2.000000,-1.000000) ;
\draw [thick] (2.000000,0.000000) --  (3.000000,0.000000) ;
\path [fill=black,opacity=.2] (2.000000,0) rectangle (3.000000,0.000000) ;
\draw [thick] (3.000000,-1.000000) --  (4.000000,-1.000000) ;
\path [fill=black,opacity=.2] (3.000000,0) rectangle (4.000000,-1.000000) ;
\draw (4.000000,-1.000000) node [right] {$\numarcs_\roadvar + \surplus_\roadvar$} ;
\end{tikzpicture}